\documentclass[11pt]{article} 

\usepackage[utf8]{inputenc} 

\usepackage{geometry} 
\usepackage{graphicx} 

\usepackage{booktabs} 
\usepackage{array} 
\usepackage{paralist} 
\usepackage{verbatim} 
\usepackage{subfig} 
\usepackage{amssymb}
\usepackage{amsmath}
\usepackage{amsthm}
\usepackage{algorithm}
\usepackage{algpseudocode}
\usepackage{fancyhdr} 
\usepackage{sectsty}
\allsectionsfont{\sffamily\mdseries\upshape} 

\usepackage[nottoc,notlof,notlot]{tocbibind} 
\usepackage[titles,subfigure]{tocloft} 

\newtheorem{theorem}{Theorem}
\newtheorem{heuristic}{Heuristic}

\title{Quantum Approaches to Learning Parity with Noise}
\author{Daniel Shiu}

\begin{document}
\maketitle

\section{Introduction}

The ``learning parity with noise'' (LPN) problem is a well-established computational challenge whose difficulty is critical to the security of several post-quantum cryptographic primitives such as HQC \cite{gaborit2025hamming} and Classic McEliece \cite{albrecht2022classic}. Classically, the best-known attacks involve information set decoding methods \cite{prange1962use}, \cite{stern1988method}, \cite{bernstein2008attacking} which are exponential in complexity for parameterisations of interest. In this paper we investigate whether quantum methods might offer alternative approaches. Our chosen line of inquiry is inspired by Regev's relating of certain lattice problems to the hidden dihedral subgroup problem \cite{regev2004quantum}. We use neighbourhoods of $\mathbb F_2^n$ to produce a function close to fulfilling Simon's promise \cite{simon1997power} with difference equal to the secret parity vector. Although unlikely to recover the secret parity vector directly, running Simon's algorithm \cite{simon1997power} produces new LPN samples. This provides the hope that we might be able to produce enough new samples to ignore one or more variables and iteratively reduce the problem.

\subsection{Learning Parity with Noise}

The learning parity with noise problem is the determination of $m$ unknown Boolean variables from $n$ samples (linear equations) where each sample is true with probability $1/2<1-p<1$. More formally, let $\mathbf s=(s_1,\ldots,s_m)^T\in \mathbb F_2^m$ be fixed and unknown. A sample is defined $(\mathbf a_i, b_i=\mathbf a_i\cdot\mathbf s\oplus e_i)$ where $\mathbf a_i=(a_{i,1},\ldots,a_{i,m})$ is a known vector sampled uniformly and independently at random from $\mathbb F_2^m$, and $e_i$ is an unknown independent Bernoulli random variable with parameter $p$. The learning parity with noise problem is to determine $\mathbf s$ with confidence, given $n$ samples.

It is natural to write a collection of samples $(a_{i,j})$ as an $n\times m$ matrix $A$ over $\mathbb F_2$ in which case we can summarise our problem as the determination of $\mathbf s$ from $A$ and $\mathbf b$ where
$$A\mathbf s\oplus\mathbf e=\mathbf b\pmod 2.$$
There is a clear similarity with the Learning with Errors (LWE) problem
$$A\mathbf s+\mathbf e=\mathbf b\pmod q$$ 
which instead works modulo $q$ and with $e_i$ drawn from a narrow probability distribution centred on 0.

The probabilistic nature of the problem means that it can only be answered with likelihood and that likelihood will be dependent on the parameters $(m,n,p)$. For certain parameters, the determination of the most likely $\mathbf s$ is believed to be hard and so the problem (again similarly to LWE) has been a popular reduction used in the design of public key cryptography schemes, particularly those of ``code-based cryptography''. The LPN problem then, is a rephrasing of the challenge of decoding a random, linear, binary code. For cryptographic applications, it can sometimes be important to consider a lightly modified version of LPN where $\mathbf e$ does not consist of independent Bernoulli variables, but instead is a uniformly random $\mathbb F_2^n$ vector of Hamming weight exactly $k$. We will often elide the distinction between the LPN problem with fixed error weight $k$ and the variable weight error with parameter $p=k/n$.

\subsection{Covering codes}

Taking inspiration from Regev \cite{regev2004quantum}, we consider a quantum super-position of our ambient space (for us this is the column space of $A$ augmented with the vector $\mathbf b$). We then aim to divide the space into neighbourhoods such that nearby elements likely belong to the same neighbourhood. The division should be such that it is computationally easy to determine a tag for each element corresponding to the neighbourhood in which it lies. If we write $f$ for the tagging function, our hope is that $f(\mathbf x)= f(\mathbf x\oplus\mathbf e)$ for a large proportion of $\mathbf x$.

A natural thought here is to use a covering code \cite{cohen1997covering}, where $f(\mathbf x)$ is some representation of the nearest element of the covering code. There are certainly covering codes for various $n$ where a nearest covering code element can be trivially computed. In later sections we examine various simple covering codes and investigate how likely the equation $f(\mathbf x)= f(\mathbf x\oplus\mathbf e)$ is to hold for various parameterisations.

\subsection{Simon's algorithm\label{sub:simon}}

If we had an ideal tagging function where $f(\mathbf x)= f(\mathbf x\oplus\mathbf e)$ for all $\mathbf x$, our problem would be easily solvable with the following quantum method due to Simon. In the following, we will sometimes treat $\mathbb F_2^m\otimes\mathbb F_2$ as column vectors over $\mathbb F_2^{m+1}$ in the obvious way. We note that, by initialising $(m+1)$ qubits to the all-zeroes state and applying a Hadamard gate to each, we have a superposition of all elements $(\mathbf u,v)$ of $\mathbb F_2^m\otimes\mathbb F_2$. With a simple circuit we can map this superposition to $A\mathbf u\oplus v\mathbf b\in\mathbb F_2^n$ in the contents of an $n$-qubit register. We can then use this register to compute $f(A\mathbf u\oplus v\mathbf b)$. Writing $F(\mathbf u,v)=f(A\mathbf u\oplus v\mathbf b)$ we have the superposition (omitting/uncomputing the $n$ intermediate qubit register and omitting normalising factor throughout).
$$\sum_{\mathbf u,v}\left|\mathbf u, v\right\rangle\left|F(\mathbf u,v)\right\rangle.$$
We  now note that if we write $(\mathbf u',v')=(\mathbf u,v)\oplus (\mathbf s, 1)$ we have $A(\mathbf u')\oplus v'\mathbf b=A\mathbf u\oplus v\mathbf b\oplus\mathbf e$ and hence (making use of our assumption on $f$) $F(\mathbf u',v')=F(\mathbf u,v)$. This recalls the ``promise'' of Simon's algorithm with difference $(\mathbf s,1)$.

Following Simon, we therefore apply Hadamard to our first $(m+1)$ qubits to obtain the superposition
$$\sum_{\mathbf w\in\mathbb F_2^{m+1}}\sum_{\mathbf u,v}(-1)^{\mathbf w\cdot(\mathbf u,v)}|\mathbf w\rangle\left|F(\mathbf u,v)\right\rangle.$$
Measuring the $F$-register\begin{footnote}{This measurement can be deferred or omitted without affecting the output. It is included to elucidate the algorithm).}\end{footnote} will return some tag value $t$ in the image space of $F$ and collapse us to the superposition
\begin{eqnarray}&&\sum_{\mathbf w\in\mathbb F_2^{m+1}}\sum_{\mathbf u,v:F(\mathbf u,v)=t}(-1)^{\mathbf w\cdot(\mathbf u,v)}|\mathbf w\rangle\left|t\right\rangle\\&=&\sum_{\mathbf w\in\mathbb F_2^{m+1}}\sum_{\mathbf u:F(\mathbf u,0)=t}\left((-1)^{\mathbf w\cdot(\mathbf u,0)}+(-1)^{\mathbf w\cdot(\mathbf u\oplus\mathbf s,1)}\right)|\mathbf w\rangle|t\rangle\\
&=&\sum_{\mathbf w\in\mathbb F_2^{m+1}}\sum_{\mathbf u:F(\mathbf u,0)=t}(-1)^{\mathbf w\cdot(\mathbf u,0)}\left(1+(-1)^{\mathbf w\cdot(\mathbf s,1)}\right)|\mathbf w\rangle|t\rangle\\
&=&2\sum_{\mathbf w:\mathbf w\cdot (\mathbf s,1)=0}\sum_{\mathbf u:F(\mathbf u,0)=t}(-1)^{\mathbf w\cdot(\mathbf u,0)}|\mathbf w\rangle|t\rangle\label{eq:measure}\
\end{eqnarray}
and measuring the first register gives some vector $\mathbf w_0$ such that $\mathbf w_0\cdot(\mathbf s,1)=0$. Repeating this $k>m$ times will give a tall matrix 
$$W=\begin{pmatrix}\mathbf w_0^T\\ \mathbf w_1^T\\ \vdots\\ \mathbf w_k^T\end{pmatrix}$$
 with $(\mathbf s,1)$ in its right nullspace. We can classically recover this nullspace using Gaussian elimination and hence $\mathbf s$.

Alert and informed readers will note that Simon's promise is usually defined as an ``if and only if'' condition, which we have not used here. It also seems unlikely that we will produce a tagging function that always keeps Simon's promise. In the next section we shall consider what happens to Simon's algorithm where the promise is not absolute.

\section{Simon's algorithm with broken promises}

In this section we consider the effects of running Simon's algorithm on a function $F:\mathbb F_2^{m+1}\to\mathbb C$ where Simon's promise ($F(\mathbf x)=F(\mathbf y)\iff \mathbf y=\mathbf x\oplus\mathbf c$ holds for some constant vector $\mathbf c$) is often, but not universally, fulfilled. We will often refer to the description in the previous section in which $\mathbf c=(\mathbf s,1)$

\subsection{Not promising ``only if''}

When we read \ref{sub:simon} with the ``only if'' condition we note that the inner summation in equation (\ref{eq:measure}) will be over a single value $\mathbf u$, and we deduce that all possible $\mathbf w$ satisfying $\mathbf w\cdot(\mathbf s,1)=0$ will be sampled equiprobably. It follows that $\mathbf w_0$ is uniformly sampled from vectors orthogonal to $(\mathbf s,1)$. Similarly, all of the rows of our tall matrix $W$ are sampled uniformly and independently at random from these orthogonal vectors and so with significant probability, after sampling $M>m$ rows we have a rank $m$ matrix.

If we omit the ``only if'' condition then the inner sum may have more than one term and there may be further cancellation or reinforcement of the phase for each $\mathbf w$. This would mean that we would not have the uniform sampling of rows of $W$. If we consider the case where there is additional structure to our function $F$ so that it is constant on cosets of the span or 2 or more vectors, then our sampled $\mathbf w_i$ will all be orthogonal to vectors in this span. Consequently, $W$ will never have rank greater than $m+1-d$ where $d$ is the dimension of the spanning set. Another way of viewing this is that Simon's algorithm is a recipe for solving the more general hidden subgroup problem for $C_2^m$, where a subgroup of size $2^d$ is recovered as the rank $2^d$ nullspace of $W$.

Of course, if $d$ is of moderate size, we can still recover the single nullspace vector $(\mathbf s,1)$, which should be distinguishable by the low Hamming weight of its associated $\mathbf e$.

For our strategy, the concern is that given $F$ which tags neighbourhoods, there may be an issue if neighbourhoods are so large that many relatively small Hamming weight vectors in the span of $A$ are likely to occur. This would give rise to $(\mathbf s,1)$ having to be recovered from a large nullspace. Note though that for Hamming weights greater than our causal error $\mathbf e$, the ``if'' part of the promise is likely to hold less often than for $\mathbf e$ and this in turn may help $\mathbf e$ be easier to recover than its heavier peers.

\subsection{Promising for a large proportion of $\mathbf x$ and $\mathbf y$}

A more significant issue is when there are cases such that $F(\mathbf x)\neq F(\mathbf x\oplus \mathbf c)$, even if this is true for a large proportion of $\mathbf x$. Exceptions prevent the complete cancellation of vectors that are not orthogonal to $\mathbf c$ and we cannot discount the possibility that $\mathbf c$ does not lie in the nullspace of our matrix of measurements. Nevertheless, we may still be able to recover information.

We divide the preimage set of a value $t_i$ in the range of $F$ into two disjoint sets $S_0$ and $S_1$ where $S_0$ is the set of paired values and $S_1$ is the set of orphaned values i.e. $\mathbf x\in S_0\Rightarrow \mathbf x\oplus\mathbf c\in S_0$ but $\mathbf x\in S_1\Rightarrow \mathbf x\oplus\mathbf c\not\in S_0\cup S_1$. If we run Simon's algorithm we get the state 
$$\sum_{\mathbf w\in\mathbb F_2^{m+1}}\sum_{\mathbf x\in\mathbb F_2^{m+1}}(-1)^{\mathbf w\cdot\mathbf x}|\mathbf w\rangle|F(\mathbf x)\rangle$$
We see that if we measure the $F$ register, we obtain $t_i$ with probability $(\#S_0+\#S_1)/2^{m+1}$ and obtain the (normalised) state
$$\frac1{\sqrt{\#S_0+\#S_1}}\sum_{\mathbf w\in\mathbb F_2^{m+1}}\sum_{\mathbf x\in S_0\cup S_1}(-1)^{\mathbf w\cdot\mathbf x}|\mathbf w\rangle|t_i\rangle.$$
We separate the state into the cases $\mathbf w\cdot\mathbf c=0$ and $\mathbf w\cdot\mathbf c=1$:
$$\frac1{\sqrt{\#S_0+\#S_1}}\left(\sum_{\mathbf w:\mathbf w\cdot\mathbf c=0}\sum_{\mathbf x\in S_0\cup S_1}(-1)^{\mathbf w\cdot\mathbf x}|\mathbf w\rangle+\sum_{\mathbf w:\mathbf w\cdot\mathbf c=1}\sum_{\mathbf x\in S_0\cup S_1}(-1)^{\mathbf w\cdot\mathbf x}|\mathbf w\rangle\right)|t_i\rangle.$$
However, 
$$\sum_{\mathbf w:\mathbf w\cdot\mathbf c=1}\sum_{\mathbf x\in S_0\cup S_1}(-1)^{\mathbf w\cdot\mathbf x}=\sum_{\mathbf w:\mathbf w\cdot\mathbf c=1}\sum_{\mathbf x\in S_1}(-1)^{\mathbf w\cdot\mathbf x}$$
due to the cancellation of the phases in $S_0$ when $\mathbf w\cdot\mathbf c=1$. Furthermore,
$$\sum_{\mathbf w:\mathbf w\cdot\mathbf c=1}|\sum_{\mathbf x\in S_1}(-1)^{\mathbf w\cdot\mathbf x}|^2\le \sum_{\mathbf w\in\mathbb F_2^{m+1}}|\sum_{\mathbf x\in S_1}(-1)^{\mathbf w\cdot\mathbf x}|^2$$
where the right hand side is $\#S_1$ by Parseval's theorem for the Walsh-Hadamard transform. We see that the probability of sampling $\mathbf w$ with $\mathbf w\cdot\mathbf c=1$ given that we have sampled $t_i$ is at most $\#S_1/(\#S_0+\#S_1)$. It follows that the probability of sampling $\mathbf w$ with $\mathbf w\cdot\mathbf c=0$ given that we have sampled $t_i$ is at least $\#S_0/(\#S_0+\#S_1)$.

Summing over all possible values of $t_i$, we can now see

\begin{theorem}\label{th:simon}
Let $\mathbf c\in\mathbb F_2^{m+1}$ be fixed and let  $F$ be a function from $\mathbb F_2^{m+1}\to \mathbb C$ where for a set $S$ of inputs $\mathbf x$ we have $F(\mathbf x)=F(\mathbf x \oplus\mathbf c)$. Then running the inner step of Simon's algorithm returns a vector $\mathbf w_0$ with $\mathbf w_0\cdot\mathbf c=0$ with probability at least $\# S/2^{m+1}$.
\end{theorem}

The theorem is slightly misleading in two different ways, which work in opposition. Firstly and helpfully, in the proof we have used the approximation
$$\sum_{\mathbf w:\mathbf w\cdot\mathbf c=1}|\sum_{\mathbf x\in S_1}(-1)^{\mathbf w\cdot\mathbf x}|^2\le \sum_{\mathbf w\in\mathbb F_2^{m+1}}|\sum_{\mathbf x\in S_1}(-1)^{\mathbf w\cdot\mathbf x}|^2$$
which is a very crude bound. Absent any further restrictions, we might expect the sum on the left hand side to only represent half of the sum on the right hand side. Heuristically then we might expect the probability quoted in theorem \ref{th:simon} to be roughly
$$\frac12+\frac{\#S}{2^{m+2}}$$
which is superior.

More subtly, as we divide into fewer neighbourhoods, the chance that $\mathbf w_0=\mathbf 0$ increases. These outcomes are counted in theorem \ref{th:simon}, but measuring $\mathbf 0$ in Simon's algorithm provides no useful information. If we have $N$ neighbourhoods and $Z_i$ vectors in each neighbourhood for $0\le i\le N-1$, then the chance of measuring $\mathbf 0$ is given by
$$\sum_{i=0}^{N-1}\frac{Z_i^2}{2^{2m+2}}.$$
Again arguing heuristically, for a uniform division into neighbourhoods, we expect $Z_i\approx 2^{m+1}/N$ which suggests that we should receive a zero measurement $1/N$ of the time. Discarding zero measurements and conditioning on non-zero measurement then gives

\begin{heuristic}\label{heu:simon}
Let $\mathbf c\in\mathbb F_2^{m+1}$ be fixed and let  $F$ be a function from $\mathbb F_2^{m+1}\to \mathbb C$ where for a set $S$ of inputs $\mathbf x$ we have $F(\mathbf x)=F(\mathbf x \oplus\mathbf c)$. Further suppose that $F$ takes $N$ distinct values each with pre-image approximately $2^{m+1}/N$ in size. Then running the inner step of Simon's algorithm is expected to return a vector $\mathbf w_0$ with $\mathbf w_0\cdot\mathbf c=0$ conditioned on $\mathbf w_0\neq\mathbf 0$ with probability at least 
$$\frac N{N-1}\left(\frac 12+\frac{\# S}{2^{m+2}}-\frac1N\right).$$
\end{heuristic}

Note that in the case $N=2$ this is the same probability as in the theorem.

\subsection{Simon's algorithm as a source of more samples}

One interpretation of theorem \ref{th:simon} with respect to our application is that applying Simon's algorithm to a suitable tagging function will produce additional samples for our LPN problem (albeit with a different noise probability $p$). Although adding to our existing collection of samples does not sound very impressive, it is potentially very powerful.

Suppose the highly optimistic situation where we can find a suitable covering code for our LPN problem, so that our Simon measurements are orthogonal to $(\mathbf s,1)$ with probability at least $p$. We could create roughly $n$ additional samples and then down select to $n$ samples which do not have $s_1$ as part of their support. This would be a LPN problem with one fewer variable but with the same number of samples and same noise rate. We could iterate our approach using Simon's algorithm on our smaller problem and so eliminate one variable at a time. After $O(mn)$ inner Simon measurements, we would have a system in $O(1)$ variables which could be easily solved, provided that we have not introduced too many other orthogonal vectors to our Simon measurement due to weakening the ``only if'' guarantee.

Even if the noise rate does grow, we can increase the number of measurements at each level and hope to offset the noise increase with a  larger set of samples.

The question now arises how well we can choose a covering code to go with our initial set of samples, so that the tagging process is likely to be constant on pairs of $(\mathbf u,v)$ that differ by the low Hamming weight vector $(\mathbf s,1)$.

\section{Best endeavours to honour Simon's promise with covering codes} 

\subsection{The repetition code}
The simplest example of a covering code is the repetition code which consists of two code words: the all zeroes and the all ones codewords. The decoding algorithm is very clear for covering $n$ dimensional space for odd $n$: we take the Hamming weight/population count of a point and assign the tag ``1'' if the value is greater than $n/2$ and the tag ``0'' if the values is less than $n/2$. For even $n$ we require a tie-breaking mechanism for the case when the weight/popcount is exactly $n/2$. We can do this for example by designating a single bit position (e.g. the least significant bit) as the tie-breaker so that points with weight $n/2$ and least significant bit 1 are tagged ``1'' and those with weight $n/2$ and least significant bit 0 are tagged ``0''. This can also be generalised to an affine repetition code where all binary strings are XORed with a fixed $n$-long constant string $\mathbf z$. The probabilistic analysis of promise-breaking is the same, though randomisation of the code might help remove certain non-causal or dependent outcomes.

We now consider the faithfulness of the repetition code to Simon's promise in the setting of a LPN problem. A single member of the column span of $A$ can be treated as a uniform random element of $\mathbb F_2^n$ with Hamming weight binomially distributed with parameters $n,1/2$. Similarly, we can treat the error vector $\mathbf e$ as having Hamming weight binomially distributed with parameters $n,p$. For convenience, we consider the simple case where $n$ is odd. The probability of selecting an element of the column span is $2^{-n}$. We divide the bits of a garbled binary string into four types: $e_0$ of the bits are fixed to zero and unchanged by the noise; $e_1$ of the bits are fixed to 1 and unchanged by the noise; $f_0$ of the bits are set to 0, but flip to 1 when noise is added; $f_1$ of the bits are set to 1 but flip when noise is added. We see that Simon's promise is broken for the repetition code for garbled strings whenever either $e_0+f_0<n/2<e_0+f_1$ or $e_0+f_0>n/2>e_0+f_1$. We can write the overall probability of promise-breaking for odd $n$ as the multinomial expression
$$\frac1{2^n}\sum_{e_0+e_1+f_0+f_2=n\atop e_0+f_0<n/2<e_0+f_1\ {\rm or}\ e_0+f_0>n/2>e_0+f_1}\binom n{e_0,e_1,f_0,f_1}p^{f_0+f_1}(1-p)^{e_0+e_1},$$
though this is not simple to estimate.

We can prove that random, light, fixed-weight garbling of a random binary string is quantifiably unlikely to change the repetition covering word, though the estimates in the proof may be too crude for practical estimates of algorithmic behaviour.

\begin{theorem}\label{th:repetition}
Given a binary string of length $n$ sampled uniformly at random, and a garbling (bit flip) of $k<n/2$ positions in that string, the original string and the garbled string are covered by the same $n$-long repetition codeword with probability at least
$$1-\frac{1+o(1)}{\frac12-\frac kn}\sqrt\frac{k\log(\pi n)}{\pi n}.$$
\end{theorem}

\begin{proof}
For our string to change from being covered by the all ones repetition word to being covered by the all zeroes repetition word, at least one of the following must be true for any value of $t$ a) either the Hamming weight of the original string lies in the range $[n/2,n/2+t]$ or b) the Hamming weight has reduced by at least $d\ge t$ for a string with original Hamming weight $n/2+d$. 

In case a) we note that the probability of the original string having a particular Hamming weight is bounded by the probability of the most likely Hamming weight $[n/2]$, for which the probability is $\binom n{[n/2]}2^{-n}\sim(\pi n)^{-1/2}$ by Stirling's approximation. Thus the probability of case a) can be bounded by $(1+o(1))t/\sqrt{\pi n}$. 

In case b) we can treat the problem as sampling and summing $k$ values from $n$ without replacement, where $n/2+d$ of ``balls in our urn'' are labelled -1 and $n/2-d$ are labelled +1. By \cite{hoeffding1963} theorem 4, we can majorise Markov's inequality for convex functions of such sums with their sampling-with-replacement counterparts. Sampling and summing with replacement gives the distribution $2\mathrm{Bin}(1/2-d/n,k)-k$ where we can apply the Chernoff-Hoeffding theorem (\cite{hoeffding1963} theorem 1).
\begin{eqnarray*}\mathbb P\left(2\mathrm{Bin}\left(\frac12-\frac dn,k\right)-k\le-d\right)&=&\mathbb P\left(\mathrm{Bin}\left(\frac12-\frac dn,k\right)-\mu<\frac {-d}2+\frac {dk}n\right)\\&\le&\exp\left(-2\frac{d^2}k\left(-\frac 12+\frac kn\right)^2\right)\end{eqnarray*}

We note that this is decreasing in $d$ for and so we can simply take the bound for $d=t$. By choosing, 
$$t=(1+o(1))\frac{\sqrt{k\log(\pi n)}}{2\left(\frac12-\frac kn\right)}.$$
We have a probability of case a) bounded by $q=(1+o(1))\sqrt{k\log(\pi n)/\pi n}/(1-2k/n)$ and a probability of case b) bounded by $o(q)$. By symmetry the probability of a change from being covered by the all zeroes repetition word to being covered by the all ones code word is the same. We conclude that the theorem as stated is true.
\end{proof}

More pragmatically, we can experiment with simulation of parameter sets of interest.

\subsubsection{Repetition code data}
An initial experiment was performed based on a toy version of the McEliece cryptographic system. A random binary Goppa code over the vector space $\mathbb F_2^{64}$ was created from a degree 7 irreducible polynomial over $\mathbb F_{64}$. This leads to a rank $m=22$ code over $\mathbb F_{2^n}$ with $n=64$ where the Hamming distance between codewords is at least 15 (so that up to 7 errors may be corrected by a user who knows the binary Goppa structure). For comparison, 1000 instances of 22 vectors in $\mathbb F_2^{64}$ were also created to make random linear codes. Likewise, 1000 instances of $2^{22}$ random $\mathbb F_2^{64}$ vectors with no \emph{a priori} linear structure were also generated.

For each of 1000 trials, a random weight 7 noise vector was created. For each codeword in the Goppa code, the repetition cover (using the least significant bit for tie-breaking) for the Goppa codeword and for the Goppa codeword plus noise were compared (likewise for each codeword (plus noise) in the random linear code and also for each of the $2^{22}$ random strings (plus noise)). The results were as follows:

\smallskip
\begin{tabular}{|c||c|c|}
\hline
Source & Covering agreements & Covering disagreements \\
\hline
Goppa code & 3143560683 ($\approx 74.95\%$) & 1050743317 ($\approx 25.05\%$)\\
Random linear code & 3157709646 ($\approx 75.28\%$) & 1036594354 ($\approx 24.72\%$)\\
Random strings & 3157447613 ($\approx 75.28\%$) & 1036856387 ($\approx 24.72\%$)\\
\hline
\end{tabular}
\smallskip

The first observation is that there is definitely a measurable bias towards agreement in all three cases. Although less pronounced, the covering code appears to be less effective for the Goppa code. This difference is not simply a question of precision: out of the 1000 trials, the Goppa code showed less agreement than the random linear code in 973 cases. This is worth noting as a principle of code-based cryptography is that randomised structured codes such as a our Goppa example should be indistinguishable from random linear codes. It is likely that this distinguishablility is due to small parameter sizes, though we remark upon it nevertheless.

Our na\"\i ve repetition code gives us hope to produce more LPN samples with a quantum computer, but of equivalent quality (the fixed weight noise roughly corresponds to LPN with $p=7/64\approx 10.94\%$). Nevertheless, we might still hope to generate many new samples, consider only samples supported on a reduced set of variables and iterate with our larger $p$. To see how the method might evolve, another 1000 trials were run with noise generated as a 64-long weight 16 vector:

\smallskip
\begin{tabular}{|c||c|c|}
\hline
Source & Covering agreements & Covering disagreements \\
\hline
Goppa code & 2666562454 ($\approx 63.58\%$) & 1527741546 ($\approx 36.42\%$)\\
Random linear code & 2653936782 ($\approx 63.27\%$) & 1540367218 ($\approx 36.73\%$)\\
Random strings & 2679290329 ($\approx 63.88\%$) & 1515013671 ($\approx 36.12\%$)\\
\hline
\end{tabular}
\smallskip

As we might expect, the degradation in sample quality continues, though there is still bias towards agreement. The distinguishability of Goppa from random reverses bias, but still seems to be causal, with Goppa worse than random in 238 out of 1000 cases. A further experiment was run with fixed weight 23 noise:

\smallskip
\begin{tabular}{|c||c|c|}
\hline
Source & Covering agreements & Covering disagreements \\
\hline
Goppa code & 2415026689 ($\approx 57.58\%$) & 1779277311 ($\approx 42.42\%$)\\
Random linear code & 2423525449 ($\approx 57.78\%$) & 1770778551 ($\approx 42.22\%$)\\
Random strings & 2423589671 ($\approx 57.78\%$) & 1770714329 ($\approx 42.22\%$)\\
\hline
\end{tabular}
\smallskip

In this case the Goppa code was worse than the random code in 660 out of 1000 cases.

\subsection{First order Reed-Muller covering codes}

A more sophisticated covering code family is the Reed-Muller family (see e.g. \cite{cohen1997covering} chapter 9). These cover a space $\mathbb F_2^n$ where $n=2^m$ by treating elements as the truth table of a Boolean function with $m$-bit inputs. Covering codewords correspond to functions supported on low-degree monomials. In particular first order Reed-Muller covering codes (a.k.a Hadamard codes) correspond to the subspace of linear functions and so the tagging algorithm is to find the single best linear approximation to the Boolean function corresponding to the element to be covered. Again, we hope that this will be resilient to the addition of noise as computation of linear approximations is linear and Parseval's theorem means that low weight codewords have a strong linear approximation to the zero function.

The tagging algorithm for a type 1 Reed-Muller code is very simple and efficient. It is a greedy version of the butterfly algorithm for computation of linear approximations where we need only track the largest contributions at each stage. Thus

\begin{algorithm}
\caption{Hadamard decoding of $\mathbf w\in\mathbb F_2^n$}
\begin{algorithmic}
\State $i\gets m-1$, $\mathrm{cover} = 0^{\otimes m}$
\While {$i\ge 0$}
  \State $\mathrm{mask}\gets (0^{\otimes 2^i}1^{\otimes 2^i})^{\otimes 2^{m-i}}$
  \State $\mathrm{right}\gets\mathbf w \& \mathrm{mask}$
  \State $\mathrm{left}\gets(\mathbf w\gg 2^i)\&\mathrm{mask}$
  \State $\mathrm{tally}\gets\mathrm{popcount}(\mathrm{right}\oplus\mathrm{left})$
  \If {$\mathrm{tally}>n/2$}
    \State $\mathrm{cover} \gets\mathrm{cover}\oplus 1\ll i$
  \Else
     \If {$\mathrm{tally}==n/2$}
       \State $\mathrm{tiebreak}$
    \EndIf
  \EndIf
  \State $i\gets i-1$
\EndWhile
\end{algorithmic}
\end{algorithm}

\noindent Again, tiebreaking for elements corresponding to functions where two or more approximations tie for the greatest strength. This could be the $2^i$th bit of $\mathbf w$ for example. The code can be extended by one bit to include the sign of approximation.

Analysis of resilience of Simon's promise is less straightforward than for the repetition code. We want to know if the strongest linear approximation of a $m$-bit Boolean function is the same as for that function XORed with a $m$-bit Boolean function that is supported on a small set. We know that the strength of approximations either reinforces or cancels linearly and that the approximations of a characteristic function of small support are bounded by Parseval's theorem. Intuitively then, we expect the coefficients of approximations not to change much and for the largest one to have a strong chance to still be largest after perturbation. Rodier \cite{rodier2003nonlinearity} has performed analysis of the largest linear approximation of a random Boolean function, which should provide a good starting point for a Reed-Muller equivalent to theorem \ref{th:repetition}. In this work, rather than delve into the distribution of approximations, we instead experiment.

\subsubsection{First order Reed-Muller data}
With the same three generation methods of elements of $\mathbb F_2^{64}$, we cover the space with $128$ neighbourhoods according to the approximation with greatest magnitude and the sign thereof. We divide the data according to the Hamming weight of the XOR of the tags of the codeword with and without noise, thus Hamming weight zero corresponds to identical tags where Simon's promise if fulfilled. As before 1000 trials were used with noise of fixed weight 7.

\smallskip
\noindent\begin{tabular}{|c||ccc|}
\hline
XOR wt. & Goppa code & Random linear code & Random words\\
\hline
0 & 494672474 ($\approx 11.79\%$) & 494017471 ($\approx 11.78\%$) & 494021786 ($\approx 11.78\%$)\\
1 & 703940515 ($\approx 16.78\%$) & 713745539 ($\approx 17.02\%$) & 713744651 ($\approx 17.02\%$)\\
2 & 1215591225 ($\approx 28.98\%$) & 1218713012 ($\approx 29.06\%$)& 1218702399 ($\approx 29.07\%$)\\
3 & 1116662541 ($\approx 26.62\%$) & 1109473397 ($\approx 26.45\%$) & 1109428675 ($\approx 26.45\%$)\\
4 & 466195684 ($\approx 11.12\%$) & 464129647 ($\approx 11.07\%$) & 464154504 ($\approx 11.07\%$)\\
5 & 171895763 ($\approx 4.10\%$) & 169166012 ($\approx 4.03\%$) & 169197526 ($\approx 4.03\%$)\\
6 & 23105647 ($\approx 0.55\%$) & 22870118 ($\approx 0.55\%$) & 22868321 ($\approx 0.55\%$)\\
7 & 2240151 ($\approx 0.05\%$) & 2188804 ($\approx 0.05\%$) & 2186138 ($\approx 0.05\%$)\\
\hline
\end{tabular}
\smallskip

The small proportion of exact matches, although significantly greater than the uncorrelated expectation of $1/128\approx 0.78\%$ will only provide a useful $\mathbf w_0$ with probability 55.54\% by heuristic \ref{heu:simon}. This makes it less useful than the two-valued repetition code. Clearly though, the tagging is capturing a great deal of information beyond the first line of the table. This might be accessible using some other approach. Perhaps more general tags that list a set of strong approximations rather than a single approximation or a quantum covering code which maps to a superposition of tags would provide better leverage.

\subsection{Choosing a covering code based on the sampling}

In the covering codes discussed, there is a great deal of scope for choosing variations that we expect to behave similarly on random elements. For the repetition code, we mentioned the affine variant where words are XORed with a fixed string before we compute the Hamming weight; for other codes we note that we expect covering properties to be preserved by permutation of the locations of the codewords. Randomisation could prove important for reducing non-causal nullspace vectors, but we should also consider whether particular variations might provide better distinguishing for our particular set of samples. The hope has some basis, given that we know in many applications there is an ideal covering code that meets Simon's promise on all pairs of codewords and light garbles. This would be the concealed error-correcting code in the McEliece cryptosystem, which is intended to be unrecoverable. Thus, we have an existence proof of an ideal cover, though one we cannot hope to use. We might hope however, that we can choose improved covers based on a limited analysis of $A$, lining up many codewords of the covering code with elements from the span of $A$ so that they lie in the centre of the covering neighbourhoods.

We recall that the repetition code seemed less effective against the Goppa code and its garbling. This might be due to a bias in the Goppa codeword towards having  Hamming weight near $n/2$. This is plausible as structurally, the Goppa codewords have stronger separation than random words or even random linear codewords. This means that although there is the zero Goppa code word, we know that there will be no Goppa codewords of weight 1-14 in our example. One might hypothesise that a weaker effect causes Goppa codewords to cluster around weight $n/2$. Translating by a random XOR of a weight $n/2$ word would presumably remove this effect. More ambitiously, if we find positions in our samples which are set on the same variables (in other words the rows of $A$ that are identical), it would make sense to choose our translation in these positions to be the same. More weakly, we might hope for value in setting the bits of our translation identically on sets of positions where the rows of $A$ are close in the Hamming metric.

For the type 1 Reed-Muller code, we note that the generator matrix is the transpose of all $\mathrm{lg} n$ binary strings listed in counting order. If we choose a subset of $\mathrm{lg} n$ columns of $A$ and sort them, we get an approximation to the counting sequence. Permuting the positions so that the value in the subset is close to the counting order will give a subspace of the span of $A$ which is close to the centres of our cover and which we may hope improves the correlation. Of course, we have many choices of $\mathrm{lg} n$ linearly independent subsets of the span of $A$ that we could use in a similar way and we might invest effort in finding one particularly close to the full counting set of strings.

\section{A small simulation \label{sec:sim}}
We conclude this paper with a simulation of how the methods of this paper might perform on a quantum computer for a toy version of a code-based cryptosystem. Taking $n=32$, using \texttt{sagemath} we generate a binary Goppa code using an irreducible degree 5 polynomial over $\mathbb F_{32}$ and the full field of $\mathbb F_{32}$ elements. This gives us a basis of $m=7$ codewords spanning a space of 128 codewords. We futher generate a random weight 5 vector of length $n=32$ corresponding to the decoding bound of our code. We then take a random $\mathbf s\in\mathbb F_2^7$, calculate the corresponding sum of basis elements and add our fixed wieght vector to create the sample vector $\mathbf b$. Code to generate the Goppa code and the output used for the simulation are given in the appendices.

It is simple to compute for all 256 values of $(\mathbf u,v)$ the $\mathbb F_2^{32}$ long vector $A\mathbf u+v\mathbf b$ and then use an untranslated 32-long repetition code to tag each vector 0 (for those with Hamming weight less than 32) or 1 (for those with Hamming weight greater than 32) with a tie-break on the least significant bit of the vector. We find that 127 of our vectors are tagged 0 and 129 are tagged 1. Of our 256 vectors, we can use guilty knowledge of $\mathbf e$ to determine that 72 ($\approx 28\%$) are tagged differently to their translate by $\mathbf e$ and 184 $\approx 72\%$ are tagged identically. We can next classically compute the Walsh-Hadmard transform on the two sets of $(\mathbf u,v)$ values according to their tag.

Due to the odd number of vectors in each tag, none of the transform coefficients are zero. There are however 135 of the 512 coefficients that are equal to $\pm 1/256$. The coefficients of the zero approximation are of course 127/256 and 129/256. Other coefficients vary in size from $\pm 3/256$ up to $\pm 39/256$. The largest coefficient on a bad approximation (one such that $\mathbf w\cdot(\mathbf s,1)=1$) is $29/256$ and the largest coefficient on a good approximation is $\pm 39/256$. Overall the probability of sampling a non-zero good approximation is $23550/65536\approx 35.9\%$ and the chance of sampling a bad approximation is $9216/65536\approx 14.1\%$, giving a proportion of good approximations around $71.8\%$ which is in agreement with the $72\%$ that we might hope for from heuristic \ref{heu:simon}. 

We use {\texttt{/dev/urandom} to randomly generate a 16-bit number and then index into the cumulative probability distribution function of the approximations to simulate the measurement. We discard zero samples, but also reject samples whose support includes either of the two least significant bits. This allows us to reduce our system to five variables. Our original system includes seven samples supported on the five most significant bits, and we generate (say) 25 new samples, hoping that the reduction in the number of variable offsets the increase in noise.

In addition to our original seven samples \begin{footnote}{Note that due to endianniess, the least signficant bit of the 6-bits corresponds to a coefficient $a_{i,3}$ of $s_3$, the next least a coefficient $a_{i,4}$ of $s_4$ and so on, with the most significant bit corresponding to the entry $b_i$ in the sample vector $\mathbf b$.}\end{footnote}
\begin{verbatim}
100001 000010 100001 000100 101000 010000 010010
\end{verbatim} 
our simulation produced the following samples:
\begin{verbatim}
111101 010001 111100 111110 100001 100011 010001 100001 
101101 111101 100001 010010 001010 100001 100000 010010 
000100 011011 011011 010010 011001 100011 100001 000110 
100001
\end{verbatim} 
We elect not to deduplicate samples as we have stronger confidence in the more likely samples.

We take a naive approach to solving the five variable system. Taking the 32-long binary strings from the five least significant bit positions:
\begin{verbatim}
10100001100111111100100011011101
01000010001010000011001011101010
00010001011000011000000100000010
00001001011000011001000011010000
00000111111001001010001011110000
\end{verbatim}
we seek a linear combination which is close to the 32-long string from the most significant bits of our samples:
\begin{verbatim}
10101001011110111100110000001101
\end{verbatim}

Exhausting all possibilities, we find that the closest combination is from the XOR of the first and fourth strings
\begin{verbatim}
10101000111111100101100000001101
\end{verbatim}
which differs in 7 places (other combinations differing in at least 9). We conclude with confidence that in our original system $s_3=1$, $s_4=0$, $s_5=0$, $s_6=1$ and $s_7=0$. With guilty knowledge, we see that this is indeed the case. Back substituting into our original samples and exhausting over the four possibilites leads us to further conclude with confidence that $s_1=1$ and $s_2=0$, again, our guilty knowledge tells us that this is indeed the case.

This small examples, give us confidence that our quantum approach offers an improvement over na\"\i ve `` exhaust and score'' methods for solving LPN. It is not clear whether more sophisticated methods such as information set decoding \cite{prange1962use}, \cite{stern1988method}, \cite{bernstein2008attacking} will benefit from the additional relations of our method as they are much more sensitive to the noise rate than exhaustion. An interesting question is whether the availability of essentially unlimited, additional, independent samples of lesser quality enables LPN attacks superior to those currently known.
\section{Summary}

We have demonstrated how applying a covering code to a superposition of the span of samples of a LPN space both with and without fixed noise creates a state with a ``best endeavours'' attempt to keep Simon's promise. We have shown that Simon's algorithm could then be used to create additional samples based on the same hidden value, but probably different noise probability. These additional samples could make the system easier to solve (e.g. by generating many new samples and discarding those not supported on a smaller set of variables). We have performed limited experimentation to show that some na\"\i ve covering codes should provide approximations to Simon's promise that might be practically useful.

We make no claim that these methods will necessarily be competitive with existing approaches, merely that they warrant deeper investigation.

\bibliographystyle{plain} 
\bibliography{qlpn} 

\appendix

\section{\texttt{sagemath} to make Goppa code}
\begin{verbatim}
F = GF(2^5)
R.<x> = F[]

g = R.random_element(5)
while g.is_irreducible() == false:
    g = R.random_element(5)

g=g/g.coefficient(5)

C = codes.GoppaCode(g,F.list())

e_set = Subsets(32,5).random_element()
e_vec=vector([0]*32)
for i in e_set:
    e_vec[i-1]=1

cw = random_vector(GF(2),7)
lpn = cw*C.generator_matrix() + e_vec

for i in range(7):
    print(hex(ZZ(list(C.generator_matrix()[i]),base=2)))

print(hex(ZZ(list(lpn),base=2)))

print("Secrets")
print("g: "+ str(g))
print("e_set: " + str(e_set))
print("e_vec: " + str(e_vec))
print("cw: " + str(cw))
\end{verbatim}

\subsection{Example output - used in section \ref{sec:sim}}
\begin{verbatim}
[1 0 0 0 0 0 1 0 0 0 1 1 1 1 1 1 1 0 0 1 1 0 1 1 1 0 0 1 0 0 0 1]
[0 1 0 0 0 0 0 0 1 0 1 0 0 0 0 0 0 1 1 1 1 1 1 1 1 1 1 1 0 1 1 0]
[0 0 1 0 1 0 0 0 1 0 0 1 0 0 1 0 1 0 0 1 0 1 1 1 0 0 0 0 0 1 1 0]
[0 0 0 1 0 0 0 0 0 0 0 0 0 1 1 1 0 0 0 0 1 1 0 0 1 0 1 1 1 1 1 1]
[0 0 0 0 0 1 1 0 0 0 1 0 1 0 1 1 0 1 1 0 1 1 1 1 1 0 1 1 0 1 0 1]
[0 0 0 0 0 0 0 1 1 0 0 0 0 1 0 0 1 1 0 1 0 0 0 1 0 1 0 1 0 1 1 1]
[0 0 0 0 0 0 0 0 0 1 1 0 0 1 1 0 1 1 1 1 0 1 1 0 0 1 0 1 1 0 1 0]
(0, 1, 1, 0, 1, 0, 1, 1, 0, 0, 1, 0, 0, 1, 0, 1, 1, 1, 0, 1, 1, 1, 1, 1, 1, 1,
 0, 0, 0, 0, 0, 0)
0x89d9fc41
0x6ffe0502
0x60e94914
0xfd30e008
0xadf6d460
0xea8b2180 
0x5a6f6600
0x3fba4d6
Secrets
g: x^5 + x^4 + (z5^3 + z5^2 + z5 + 1)*x^3 + (z5^4 + z5)*x^2 + 
(z5^4 + z5^3 + z5 + 1)*x + z5^4 + 1
e_set: {1, 2, 13, 14, 23}
e_vec: (1, 1, 0, 0, 0, 0, 0, 0, 0, 0, 0, 0, 1, 1, 0, 0, 0, 0, 0, 0, 0, 0, 1, 0, 
0, 0, 0, 0, 0, 0, 0, 0)
cw: (1, 0, 1, 0, 0, 1, 0)
\end{verbatim}

\section{C code to simulate Simon's algorithm sampling}
\begin{verbatim}
#include <stdint.h>
#include <stdio.h>
#include <stdlib.h>
#include <nmmintrin.h>

uint32_t goppa[]={0x89d9fc41,                         /* Hard-coded
                                            Goppa code from sagemath */
0x6ffe0502,
0x60e94914,
0xfd30e008,
0xadf6d460,
0xea8b2180,
0x5a6f6600,
0x3fba4d6};

uint32_t e = 1 ^ 1<<1 ^ 1<<12 ^ 1<<13 ^ 1<<22;        /* Error vector 
                                            used in generating 0x3fba4d6 */
uint32_t s = 0xA5;                                    /* Secret vector 
                                            used in generating 0x3fba4d6 */

uint8_t repetition_cover(uint32_t word32) {           /* Untranslated 
                                            repetition code */
 uint32_t tally;
 uint8_t cover_tag = 0;
 
 tally = __builtin_popcount(word32);
 if (tally > 16)                                      /* Tag high 
                                            popcounts with 1 */
  cover_tag = 1;
 if (tally ==16) 
  cover_tag = word32 & 1;                             /* Tie break balanced 
                                            words using least significant bit */
 
 return cover_tag;
}

int main() {
 uint32_t gcodeword, index; //, type0=0, type1=0;
 int32_t indices[2][256]={0};
 uint32_t i, j, k, total, trials, sample, good=0;
 FILE * fp = fopen("/dev/urandom","r");

 if (fp==NULL) {
  printf("Failed to open /dev/urandom\n");
  exit(-1);
 }

 gcodeword = 0;
 index = 0;
 for(j=0;j < 256;j++) { 
  indices[repetition_cover(gcodeword)][index]++;
  gcodeword ^= goppa[__builtin_ctz(j)];               /* Gray code stepping */
  index ^= 1<<__builtin_ctz(j);
 }

 for(i=128;i>0;i=i>>1)                                /* Butterfly computation 
                                            of the Walsh-Hadamard transform */
  for(j=0;j<128/i;j++)
   for(k=0;k<i;k++) {
    indices[0][2*i*j+k] = indices[0][2*i*j+k] + indices[0][2*i*j+i+k];
    indices[0][2*i*j+i+k] = indices[0][2*i*j+k] - 2*indices[0][2*i*j+i+k];
    indices[1][2*i*j+k] = indices[1][2*i*j+k] + indices[1][2*i*j+i+k];
    indices[1][2*i*j+i+k] = indices[1][2*i*j+k] - 2*indices[1][2*i*j+i+k];
    } 

 total = 0;                                           /* Convert the Hadamard 
                                            coefficients into sample probability 
                                            numerators */
 for(i=0;i<256;i++){
  indices[0][i] = indices[0][i]*indices[0][i];
  indices[0][i] += indices[1][i]*indices[1][i];       /* Combine the sample 
                                            probability numerators from both 
                                            tags */
  // total += indices[0][i];                          /* Sanity check that 
                                            numerators sum to 2^{2(m+1)} */
  // if (__builtin_popcount(i&s)%2==0)
  //  type0 += indices[0][i];                         /* Use guilty knowledge 
                                            to check probability of sampling 
                                            a Simon sample */
  // else 
  //  type1 += indices[0][i];                         /* Use guilty knowledge 
                                            to check probabiltiy of sampling 
                                            a non-Simon sample */
 }

 total = 0;
 good = 0;
 for(trials = 0; trials <500; trials++) {             /* Simulate measurement 
                                            using our CDF */
  sample = (uint32_t)fgetc(fp);
  sample = sample<<8 ^ (uint32_t)fgetc(fp);           /* Random 2^16 value */
  i=0;                                                /* Start with 
                                            coefficient 0 with cumulative 
                                            probability 0 */
  index = 0;
  while (i<256){
   index += indices[0][i];                            /* Increment CDF 
                                            by probability of measuring index */
   if (index > sample) break;                         /* Select the case 
                                            where the cumulative probability 
                                            exceeds the random value */
   i++;
  }
  if ((i!=0)&&((i&3)==0)) {                           /* Discards 0 samples 
                                            and samples whose support 
                                            overlaps with the 2 LSBs */
   total++;                                              
   if (__builtin_popcount(i&s)%2==0) good++;          /* Use guilty knowledge
                                             to track the number of good samples */
   printf("Sample %x trials %d total %d good %d\n", i, trials, total, good);
  } 
  if (total==25) break;                               /* Hard wired to stop 
                                            after getting 25 new samples */
 } 

 fclose(fp);
 return 0;
}

\end{verbatim}

\end{document}